\documentclass[12pt]{iopart}
\usepackage{iopams}
\usepackage{amstext}
\usepackage{amsthm}
\usepackage{dsfont}
\usepackage{hyperref}

\newcommand{\eqref}[1]{(\ref{#1})} 

\newcommand{\sL}{\mathcal{L}}
\newcommand{\RR}{\mathbb{R}}
\newcommand{\NN}{\mathbb{N}}
\newcommand{\EE}{\mathbb{E}}
\newcommand{\var}{\mathrm{var}}
\newcommand{\PP}{\mathbb{P}}
\newcommand{\ZZ}{\mathbb{Z}}
\newcommand{\walklength}{\sL_n}
\renewcommand{\d}{\mathrm{d}}
\newcommand{\ind}{{\mathds{1}}}

\newcommand{\floor}[1]{\lfloor #1 \rfloor}

\newtheorem{theorem}{Theorem}[section]
\newtheorem{lemma}[theorem]{Lemma}
\newtheorem{remark}[theorem]{Remark}

\newtheorem{proposition}[theorem]{Proposition}

\begin{document}
\title{The length of self-avoiding walks on the complete graph}
\author{Youjin Deng$^{1,2}$, Timothy M Garoni$^3$, Jens Grimm$^3$, Abrahim Nasrawi$^3$ and Zongzheng Zhou$^3$}
\address{$^1$ Department of Modern Physics, University of Science and Technology of China, Hefei, Anhui 230026, China}
\address{$^2$ National Laboratory for Physical Sciences at Microscale, University of Science and Technology of China, Hefei, Anhui 230026, China}
\address{$^3$ ARC Centre of Excellence for Mathematical and Statistical Frontiers (ACEMS),
School of Mathematics, Monash University, Clayton, Victoria 3800, Australia}
\ead{\mailto{yjdeng@ustc.edu.cn}, \mailto{tim.garoni@monash.edu}, \mailto{jens.grimm@monash.edu}, \mailto{abrahim.nasrawi@monash.edu}, 
  \mailto{eric.zhou@monash.edu}}
\vspace{10pt}
\begin{indented}
\item[]14 June 2019
\end{indented}

\begin{abstract}
We study the variable-length ensemble of self-avoiding walks on the complete graph.  We obtain the leading order asymptotics of the mean and
variance of the walk length, as the number of vertices goes to infinity. Central limit theorems for the walk length are also established, in
various regimes of fugacity. Particular attention is given to sequences of fugacities that converge to the critical point, and the effect of
the rate of convergence of these fugacity sequences on the limiting walk length is studied in detail. Physically, this corresponds to
studying the asymptotic walk length on a general class of pseudocritical points.
\end{abstract}


\section{Introduction}
The self-avoiding walk (SAW) is a fundamental and well-studied model in discrete statistical mechanics~\cite{MadrasSlade96}.  An $n$-step
SAW on a graph $G$ is simply a path on $G$ of length $n$; i.e., a sequence of $n+1$ distinct vertices, such that each pair of consecutive
vertices are adjacent.  By far the most well-studied choice of $G$ is the infinite lattice $\ZZ^d$, where the focus is on SAWs rooted at the
origin. The fixed-length (canonical) ensemble considers uniformly random SAWs of given length, while the variable-length (grand canonical)
ensemble weights SAWs with a fugacity for the number of steps.  

For $d\ge5$, many properties of SAWs on $\ZZ^d$ can be established rigorously via the lace expansion~\cite{HaraSlade1992,Hara2008}, which
show that SAWs behave in many ways like simple random walks (SRWs) in this case. Moreover, the behaviour of other statistical mechanical
models on $\ZZ^d$, in particular the Ising model of ferromagnetism~\cite{FriedliVelenik2017}, have also been found to coincide with SAW/SRW
behaviour in sufficiently high dimension; for example, all three models display the same asymptotics of their two-point
functions~\cite{Sakai2007}.

By contrast with the infinite lattice, comparatively few results have been obtained for SAWs on finite graphs; see~\cite{Yadin2016} and
references therein. In particular, it seems that the first study of SAWs on the complete graph was only undertaken very
recently~\cite{Yadin2016}. This is in stark contrast to other models in discrete statistical mechanics, such as the Ising model and
percolation, where the complete graph models, corresponding to the Curie-Weiss model and Erd\H{o}s-R\'enyi random graphs, respectively, are
classical topics~\cite{FriedliVelenik2017,JansonLuczakRucinski2000}.

The complete graph on $n$ vertices, $K_n$, is a simple undirected graph in which each pair of distinct vertices are adjacent. One motivation for studying statistical mechanical models on the complete graph is that one expects the large $n$ asymptotic behaviour
of the model on $K_n$ to coincide with the large $L$ behaviour\footnote{More precisely, one compares $n$ with $L^d$.}  of the model on the
torus $\ZZ_L^d$, for all sufficiently large $d$.  Physically, $\ZZ_L^d$ simply corresponds to a box of side-length $L$, with periodic boundary conditions.

In contrast to the well-understood behaviour of the SAW and Ising models on the infinite lattice $\ZZ^d$ when $d\ge5$, their finite-size
scaling behaviour on finite boxes appears to be surprisingly
subtle, and has been the subject of longstanding debate~\cite{Brezin1982,Fisher1983,BinderNauenbergPrivmanYoung1985,Binder1985,RickwardtNielabaBinder1994,KennaBerche2014,WittmannYoung2014,FloresSolaBercheKennaWeigel2016,GrimmElciZhouGaroniDeng2017,Grimm2018,ZhouGrimmFangDengGaroni2018}.
In particular, there is now strong numerical evidence that the critical scaling of the SAW and Ising two-point functions is boundary
dependent for $d\ge5$: on boxes with free boundary
conditions~\cite{LundowMarkstrom2011,WittmannYoung2014,LundowMarkstrom2014,GrimmElciZhouGaroniDeng2017,Grimm2018,ZhouGrimmFangDengGaroni2018},
the critical scaling agrees with that observed on the infinite lattice, whereas with periodic boundary
conditions~\cite{Grimm2018,ZhouGrimmFangDengGaroni2018}, the scaling is of an anomalous form first conjectured by Papathanakos~\cite{Papathanakos2006}.

Similarly, it was observed numerically~\cite{ZhouGrimmFangDengGaroni2018} that
the mean walk length of SAW, restricted to a finite box in $\ZZ^d$, also depends strongly on the boundary conditions imposed. However,
as argued in~\cite{Grimm2018,ZhouGrimmFangDengGaroni2018}, there is strong numerical and theoretical evidence to suggest that
the SAW two-point function is \emph{only} affected by boundary conditions \emph{via} their effect on the mean walk length.

Specifically, it was proposed in~\cite{Grimm2018,ZhouGrimmFangDengGaroni2018} that the SAW and Ising two-point functions on high-dimensional
finite boxes should coincide with the two-point function of a simple random walk, just as occurs on the infinite lattice $\ZZ^d$, provided
the random walk has an appropriately chosen random (finite) length. It can be shown
rigorously~\cite{Grimm2018,ZhouGrimmDengGaroni2019,GrimmZhouDengGaroni2019} that the asymptotics of the two-point function of this
random-length random walk (RLRW) on a finite box\footnote{With dimension $d\ge3$.} is determined entirely by the walk length distribution; for given walk length asymptotics, the model predicts the same two-point function asymptotics for periodic and free boundary conditions.

As a first test of this RLRW scenario, it was observed~\cite{ZhouGrimmFangDengGaroni2018} that inputting the asymptotic walk length for the
critical complete-graph SAW~\cite{Yadin2016} into the RLRW model (on boxes with either free or periodic boundary conditions) correctly
predicts the Papathanakos scaling of the two-point function. However, this RLRW correspondence is seen to hold not only at the
thermodynamic critical point, but also at general pseudocritical points. For example, the RLRW picture predicts, and numerical evidence
strongly suggests~\cite{Grimm2018,ZhouGrimmDengGaroni2019,GrimmZhouDengGaroni2019}, that the Papathanakos scaling of the two-point function
can be observed with free boundary conditions, at an appropriate pseudocritical point, thus clarifying a recent
debate~\cite{BercheKennaWalter2012,WittmannYoung2014,LundowMarkstrom2016} on the matter.

Moreover, by studying the complete-graph SAW on a general family of pseudocritical points, it was shown~\cite{ZhouGrimmFangDengGaroni2018}
that a continuous family of scalings for the mean walk length was possible, and numerical
evidence~\cite{Grimm2018,ZhouGrimmFangDengGaroni2018} strongly suggested these results extend to tori with $d\ge5$. When inputted into the
RLRW model, these asymptotic mean lengths then predict a continuous family of possible anomalous scalings of the two-point function on the torus,
providing a broad generalisation of the Papathanakos conjecture. As a special case, this theory then predicts that the standard infinite
lattice scaling of the two-point function \emph{can} be observed on tori, at a certain (explicitly characterised) pseudocritical
point. These predictions for the two-point function are all strongly supported by simulations, both for SAW and for the Ising
model~\cite{Grimm2018,ZhouGrimmFangDengGaroni2018,ZhouGrimmDengGaroni2019}.

No proof was given in~\cite{ZhouGrimmFangDengGaroni2018} for the stated asymptotics of the mean walk length on the complete graph, and one
purpose of the current paper is to rigorously justify these claims. However, we go further. Motivated by the above applications, we provide
a systematic study of the walk length of the complete-graph SAW, for a variety of regimes of fugacity. We derive the leading order
asymptotics of both the mean and variance, and establish non-degenerate central limit theorems in each regime.

The outline of the remainder of the paper is as follows. In Section~\ref{subsec:results} we define the model more precisely, and state our
results. Section~\ref{sec:preliminaries} shows how to express the moments and distribution function of the walk length in terms of incomplete gamma functions, and summarises the known asymptotic behaviour of the latter. 
Section~\ref{sec:mean and variance} then proves the stated asymptotics for the mean and variance, while Section~\ref{sec:limit theorems} proves the stated limit theorems.

\subsection{Results}
\label{subsec:results}
Let $K_n$ denote the complete graph with vertex set $V(K_n) = \{0,1,2,\cdots,n-1\}$.
Because every pair of vertices in $K_n$ are adjacent, the set of $k$-step self-avoiding walks on $K_n$ starting at 0 is simply
\begin{equation}
 \Omega_{n,k} := \{(\omega_0,\omega_1,\cdots,\omega_k)\in V^{k+1}(K_n): \omega_0=0,\omega_i\neq\omega_j~\,\forall\,~ i\neq j\}
\end{equation}
It is easily seen that the number of such walks is then
\begin{equation}
|\Omega_{n,k}| = (n-1)(n-2)\cdots(n-k) = \frac{(n-1)!}{(n-k-1)!}
\label{eq:size of Omega_k}
\end{equation}
We denote by $\Omega_n:=\bigcup_{k=0}^{n-1}\Omega_{n,k}$ the set of all self-avoiding walks (starting from 0) on $K_n$.
The \emph{length} of a SAW, denoted by $\walklength:\Omega_n\rightarrow\NN$, is simply its number of steps, so that
$\walklength(\omega)=k$ for every $\omega\in\Omega_{n,k}$, and $0\le \walklength(\omega)\le n-1$ for every $\omega\in\Omega_n$.

For given $z>0$, we define the corresponding (variable-length) self-avoiding walk model on $K_n$ via
\begin{equation}
\label{eq:SAW measure}
  \PP_{n,z}(\omega) = 
  \frac{(z/n)^{\sL_n(\omega)}}{\sum_{\omega'\in\Omega_n}(z/n)^{\sL_n(\omega')}} \ \qquad\,\text{ for all }\, \ \omega\in\Omega_n
\end{equation}
As would be expected by analogy with the Curie-Weiss and Erd\H{o}s-R\'enyi models, we have scaled the fugacity by $n$ in the above
definition. This choice of scaling can be justified post facto via the non-trivial results it leads to in~\cite{Yadin2016}, and
Theorems~\ref{Thm:mean_variance} and~\ref{Thm:CLT} below. We denote expectation and variance with respect to $\PP_{n,z}$ by $\EE_{n,z}$ and $\var_{n,z}$. 

\begin{theorem}
\label{Thm:mean_variance}
Fix $\alpha\in\RR$ and $\beta\geq 0$, and define the sequence $(z_n)_{n\in\NN}$ so that ${1/z_n=1+\alpha~ n^{-\beta}}$. Let $X$ denote a standard normal random variable. Then, as $n\to\infty$
\begin{equation*}
\EE_{n,z_n}(\walklength) 
\sim 
\begin{cases}
 |\alpha|n,  & \beta=0, \alpha\in(-1,0) \\[5pt]
 |\alpha|n^{1-\beta},  & \beta\in(0,1/2), \alpha<0 \\[5pt]
 \left[\EE(X|X>\alpha) - \alpha\right]\sqrt{n},  & \beta=1/2,\alpha\in\RR \\[5pt]
 \sqrt{2/\pi} \,\sqrt{n},  & \beta>1/2,\alpha\in\RR \\[5pt]
 n^{\beta}/\alpha,  & \beta\in(0,1/2), \alpha>0 \\[5pt]
 1/\alpha,  & \beta=0, \alpha>0 \\[5pt]
\end{cases}
\end{equation*}
and
\begin{equation*}
\var_{n,z_n}(\walklength) 
\sim 
\begin{cases}
(1+\alpha)n & \beta=0, \alpha\in(-1,0) \\[5pt]
n & \beta\in(0,1/2), \alpha<0 \\[5pt]
\var(X|X>\alpha)\,n 
& \beta=1/2,\alpha\in\RR \\[5pt]
(1-2/\pi)n & \beta>1/2,\alpha\in\RR \\[5pt]
n^{2\beta}/\alpha^2 & \beta\in(0,1/2), \alpha>0 \\[5pt]
(1+\alpha)/\alpha^2 & \beta=0, \alpha>0 \\[5pt]
\end{cases}
\end{equation*}
 \end{theorem}
Note that if we denote the standard normal density and tail distribution by 
\begin{equation}
\varphi(x):=\frac{e^{-x^2/2}}{\sqrt{2\pi}},
\qquad \qquad \qquad 
\bar{\Phi}(x):=\int_{x}^{\infty} \varphi(t) \mathrm{d}t
\label{eq:normal distribution}
\end{equation}
then we have the following explicit forms for the conditional Gaussian moments appearing above
\begin{equation}
\label{eq:conditional normal moments}
\hspace{-10mm}
\EE(X|X>\alpha) = \frac{\varphi(\alpha)}{\bar{\Phi}(\alpha)}, \qquad \quad
\var(X|X>\alpha) =1+\frac{\varphi(\alpha)}{\bar{\Phi}(\alpha)}\left(\alpha-\frac{\varphi(\alpha)}{\bar{\Phi}(\alpha)}\right)
\end{equation}

It is also possible to obtain central limit theorems for $\sL_n$, in the various fugacity regimes. We have:
\begin{theorem}
\label{Thm:CLT}
Fix $\alpha\in\RR$ and $\beta\geq 0$, and define the sequence $(z_n)_{n\in\NN}$ so that ${1/z_n=1+\alpha~ n^{-\beta}}$. 
Let $X$ denote a standard normal random variable, $Y$ an exponential random variable with mean 1, and, if $\alpha>0$,
let $W_{\alpha}$ denote a geometric random variable with mean $1+1/\alpha$. As $n\to \infty$
\begin{enumerate}
    \item\label{CLT:low temp} If $\beta=0$ and $\alpha\in(-1,0)$, then 
$$
\frac{\sL_n - |\alpha|n}{\sqrt{(1+\alpha)n}} \Longrightarrow X
$$
\item\label{CLT:low temp window} If $\beta\in(0,1/2)$ and $\alpha<0$, then 
$$
\frac{\sL_n - |\alpha|n^{1-\beta}}{\sqrt{n}} \Longrightarrow X
$$
    \item\label{CLT:boundary} If $\beta=1/2$ and $\alpha\in\RR$, then
$$
        \frac{\sL_n}{\sqrt{n}} + \alpha \Longrightarrow X \,|\, X > \alpha
$$
    \item\label{CLT:critical window} If $\beta >1/2$ and $\alpha\in\RR$, then
$$
  \frac{\sL_n}{\sqrt{n}}\Longrightarrow X \,|\, X > 0
$$
    \item\label{CLT:high temp window} If $\beta\in(0,1/2)$ and $\alpha>0$ then,
$$
\frac{\walklength}{n^\beta/\alpha} \Longrightarrow Y
$$
    \item\label{CLT:high temp} If $\beta = 0$ and $\alpha >0$, then
$$
\sL_n + 1 \Longrightarrow W_{\alpha}
$$
\end{enumerate}
\end{theorem}
\begin{remark}
The notation used in Case~\eqref{CLT:boundary} of Theorem~\ref{Thm:CLT} is an abbreviation for the following statement:
$$
\lim_{n\to\infty} \PP_{n,z_n}\left(\frac{\sL_n}{\sqrt{n}}+\alpha \le  y \right) = \PP(X\le y | X>\alpha), \qquad \text{ for all } y\in\RR,
$$
and Case~\eqref{CLT:critical window} is to be interpreted analogously. In addition, since the distribution of $X$ conditioned on $X>0$
equals the distribution of $|X|$, Case~\eqref{CLT:critical window} can equivalently be interpreted as stating that the law of
$\sL_n/\sqrt{n}$ converges weakly to a standard half-normal distribution.
\end{remark}

The above asymptotics of the mean in the case $\alpha,\beta>0$ were announced previously in~\cite{ZhouGrimmFangDengGaroni2018}.  The main
technical tool used in the proof of Theorems~\ref{Thm:mean_variance} and~\ref{Thm:CLT} is a uniform asymptotic expansion of the incomplete
gamma function~\cite{Temme1979}. During the final stages of preparation of this article, we became aware of the very recent manuscript~\cite{Slade2019},
which also studies SAW on the complete graph. The asymptotic expressions for the mean given in Theorem~\ref{Thm:mean_variance} are also presented in~\cite{Slade2019}. In addition, the limit theorems given in Theorem~\ref{Thm:CLT} in Cases~\eqref{CLT:high temp window} and~\eqref{CLT:high temp} agree with Eqs.~(1.27) and~(1.26) from~\cite{Slade2019}. In Case~\eqref{CLT:boundary}, our limiting distribution expressed in terms of a conditioned normal variable appears superficially different to \cite[Eq.~(1.28)]{Slade2019}, which is characterised in terms of a moment generating function,
however it can be easily verified that the results are equivalent. In Cases~\eqref{CLT:low temp} and~\eqref{CLT:low temp window}, our standardisation of $\sL_n$ yields non-degenerate limits which are not equivalent to the limits presented in Eqs.~(1.30) and~(1.29) of~\cite{Slade2019}.

\section{Preliminaries}
\label{sec:preliminaries}
In this section we show how to express the mean, variance and tail distribution of $\sL_n$ in terms of incomplete gamma functions, and then
summarise the known asymptotic behaviour of the latter. 

Let $\nu:=n/z$. As observed in~\cite{Yadin2016}, and as can be easily verified directly from~\eqref{eq:size of Omega_k} and~\eqref{eq:SAW measure}, the random variable $n-1-\sL_n$ has the distribution of a rate-$\nu$ Poisson random variable conditioned on being less than $n$. It follows that, for all integer $0\le k \le n-1$, we have
\begin{equation}
    \PP_{n,z}(\sL_n  = n - 1 - k) = \frac{e^{-\nu}\, \nu^{k}}{k! \,Q(n,\nu)}
    \label{eq:mass function identity}
\end{equation}
where $Q(n,\nu)$ denotes the regularised incomplete gamma function~\cite{Temme1996}. The function $Q(n,\nu)$ can be defined for all real $n,\nu>0$ by
\begin{equation}
    Q(n,\nu) := \frac{1}{\Gamma(n)}\int_{\nu}^{\infty} s^{n-1}\,e^{-s}\,\d s,
    \label{eq:Q as integral}
\end{equation}
and in obtaining~\eqref{eq:mass function identity} we have used the fact that, for any positive integer $n$, integration by parts yields
\begin{equation}
\label{eq:Q as a sum}
    e^{-\nu} \sum_{j=0}^{n-1}\frac{\nu^j}{j!} = Q(n,\nu)
\end{equation}
Now, for any $x\in\RR$
$$
\PP_{n,z}(\sL_n > x) = \ind(x<0) + \ind(x\geq 0) \sum_{k=0}^{n-\floor{x}-2}\PP_{n,z}(\sL_n = n-1-k)
$$
and so applying~\eqref{eq:mass function identity} and~\eqref{eq:Q as a sum} then yields the following expression for the tail distribution
\begin{equation}
\label{eq:tail_distribution}
\PP_{n,z}(\sL_n > x) =  \ind(x<0) + \ind(x\geq 0) \frac{Q(n-\floor{x}-1,\nu)}{Q(n,\nu)}
\end{equation}

Similarly, from~\eqref{eq:mass function identity} and~\eqref{eq:Q as a sum} we can see that for any $m\in\NN$, the $m$th moment of $\sL_n$ is
\begin{eqnarray}
\EE_{n,z}(\sL_n^m) &= \sum_{k=0}^{n-1} (n-1-k)^m \PP_{n,z}(\sL_n=n-1-k) \nonumber \\
&= \sum_{l=0}^m (-1)^l {m \choose l}(n-1)^{m-l} \frac{e^{-\nu}}{Q(n,\nu)}\sum_{k=0}^{n-1}k^l\frac{\nu^k}{k!} \nonumber \\
&= \sum_{l=0}^m (-1)^l {m \choose l}(n-1)^{m-l} \frac{e^{-\nu}}{Q(n,\nu)} \left(\nu \frac{\d}{\d \nu}\right)^l 
e^{\nu}Q(n,\nu)
\label{eq:mth moment}
\end{eqnarray}
The differentiations required in~\eqref{eq:mth moment} are easily performed using~\eqref{eq:Q as integral}, and one obtains
\begin{eqnarray}
\label{Mean and Variance}
\EE_{n,z}(\sL_n) &= n-1-\nu +\frac{1}{H_n(\nu)} \label{eq:mean in terms of H}\\
\var_{n,z}(\sL_n) &= \nu + \frac{(\nu-n)}{H_n(\nu)} - \frac{1}{H_n^2(\nu)}\label{eq:var in terms of H}
\end{eqnarray}
where
\begin{eqnarray}
\label{Def:H_n}
H_n(\nu):=\frac{\Gamma(n)\,Q(n,\nu)}{\nu^n\,e^{-\nu}}
\end{eqnarray}

It will be useful in what follows to observe that if $\nu=n \lambda_n$ for some sequence $(\lambda_n)_{n\in\NN}$, then $H_n(n\lambda_n)$ can be written
\begin{equation}
\label{eq:H in terms of eta}
H_n(n\lambda_n) = \sqrt{\frac{2\pi}{n}} \exp\left(n \,\eta^2(\lambda_n)/2\right)\,\Gamma^*(n)\,Q(n,n\lambda_n)
\end{equation}
where 
\begin{equation}
\label{eq:definition of gamma star}
\Gamma^*(n):=\sqrt{\frac{n}{2\pi}}\,e^n\,n^{-n}\,\Gamma(n)    
\end{equation}
and $\eta:(0,\infty)\to\RR$ is defined by
\begin{equation}
\label{eq:eta definition}
\eta(\lambda) = (\lambda - 1)\sqrt{2\frac{\lambda-1-\log \lambda}{(\lambda-1)^2}}
\end{equation}
We also note, for future reference, that
\begin{equation}
    \label{eq:eta Taylor expansion near 1}
\eta(\lambda) = (\lambda-1) + O(\lambda-1)^2, \qquad \lambda \to 1
\end{equation}

\subsection{Gamma function asymptotics}
We collect here some known asymptotic results, which will be used in our proofs of Theorems~\ref{Thm:mean_variance} and~\ref{Thm:CLT}. To begin, we note that Stirling's formula (e.g.~\cite[Eq.(6.1.15)]{Temme2015}) implies that as $n\to\infty$
\begin{equation}
\label{eq:Stirling's formula}
    \Gamma^*(n) = 1 + \frac{1}{12n} + O(n^{-2})
\end{equation}

We next state three useful expansions for $Q$. The first follows from~\cite[Eqs.(7.4.41) and (7.4.43)]{Temme2015}, which dates back to Tricomi~\cite{Tricomi1950}.
\begin{lemma}[Tricomi]
\label{lem:Tricomi's expansion}
Let $\lambda>1$, and let $(\theta_n)_{n\in\NN}$ satisfy $\theta_n \rightarrow 1$ as $n\rightarrow \infty$. Then, 
for any $N\in\NN$, as $n\to\infty$
$$
\Gamma(n\theta_n)\,Q(n\theta_n,n\lambda) = (n\lambda)^{n\theta_n-1} e^{-n\lambda}
\left[\sum_{k=0}^{N-1}
b_k(\theta_n/\lambda)(\lambda\,n)^{-k}
+ O(n^{-N})\right]
$$
where the first few coefficients are 
$$
b_0(x) = \frac{1}{1-x}, \qquad b_1(x) = \frac{-1}{(1-x)^3}, \qquad b_2(x) = \frac{x+2}{(1-x)^5}
$$
\end{lemma}

The following two expansions, which are the main ingredients in the proofs of Theorem~\ref{Thm:mean_variance} and~\ref{Thm:CLT}, are consequences of a uniform asymptotic expansion due to Temme~\cite{Temme1979}. The functions $\bar{\Phi}$ and $\eta$ appearing below are as defined above in~\eqref{eq:normal distribution} and~\eqref{eq:eta definition}, respectively.
\begin{lemma}[Temme]
\label{lem:Temme's expansion}
Let $(\theta_n)_{n\in\NN}$ and $(\lambda_n)_{n\in\NN}$ denote convergent real sequences with positive limits. Then, as $n\to\infty$
$$
Q(n\theta_n,n\lambda_n) = \bar{\Phi}(\sqrt{n\theta_n}\eta(\lambda_n/\theta_n)) + O\left(\frac{e^{-n\theta_n\eta^2(\lambda_n/\theta_n)/2}}{\sqrt{ n}}\right)
$$
\end{lemma}

Since $\bar{\Phi}(x)\to0$ as $x\to\infty$, if the sequences $\theta_n$ and $\lambda_n$ are such that $\sqrt{n}\eta(\lambda_n/\theta_n) \rightarrow \infty$, then we require a stronger result than Lemma~\ref{lem:Temme's expansion}. The following is also a consequence of Temme's uniform asymptotic expansion~\cite{Temme1979}.
\begin{lemma}[Temme]
\label{lem:Temme's expansion after cancellation}
Let $(\theta_n)_{n\in\NN}$ and $(\lambda_n)_{n\in\NN}$ denote real sequences, each converging to 1, 
let $\eta_n:=\eta(\lambda_n/\theta_n)$ and $\xi_n := \frac{\lambda_n}{\theta_n} - 1$, and suppose
$\lim\limits_{n\to\infty}\sqrt{n} \eta_n = +\infty$. Then, as $n\to\infty$
$$
\sqrt{2\pi n\theta_n}e^{ n\theta_n \eta_n^2/2}Q(n\theta_n,n\lambda_n)
=
\sum_{k=0}^{N-1}(-1)^k
\frac{q_k(\xi_n)}{(n\theta_n)^k\xi^{2k+1}_n} + O\left(\frac{1}{\eta_n^{2N+1}\,n^N}\right)
$$
where the $q_k(x)$ are polynomials of degree $2k$,
$$
q_k(x):=\sum_{l=0}^{2k} a_{k,l}\, x^l
$$
and $a_{0,0} = a_{1,0} = a_{1,1} = 1\;, \ a_{1,2} = 1/12\;, \ a_{2,0} = 3$.
\end{lemma}
\begin{proof}
Combining Eqs. (1.4), (1.5) and (3.3) from~\cite{Temme1979} implies that, for any $N\in\NN$, as $n\to\infty$
$$
Q(n\theta_n,n\lambda_n) = \bar{\Phi}(\sqrt{n\theta_n}\eta_n) + 
\frac{e^{-n \theta_n \eta_n^2/2}}{\sqrt{2\pi  n\theta_n}}
\left[
\sum_{k=0}^{N-1} 
\frac{(-1)^{k}}{(n\theta_n)^k}
\!\!\left(\frac{q_k(\xi_n)}{\xi_n^{2k+1}}- \frac{A_k}{\eta_n^{2k+1}}\right)
+ O\left(n^{-N}\right)
\right]
$$
where
$$
A_k:=2^k \frac{\Gamma(k+1/2)}{\Gamma(1/2)}
$$
The construction of the polynomials $q_k(x)$ is discussed in~\cite[Eqs.~(3.5)-(3.7)]{Temme1979}, and
the values of the coefficients $a_{k,l}$ stated above can be obtained directly from~\cite[Eq.~(3.4)]{Temme1979}. 

Now, the well-known asymptotic expansion of $\bar{\Phi}(z)$ as $z\to+\infty$ (see e.g.~\cite[Eq.~(3.4.30)]{Temme2015}) implies that for all $M\in\NN$ we have that as $n\to\infty$
\begin{equation}
\label{eq:large argument Phi expansion}
\hspace{-15mm}
\bar{\Phi}(\sqrt{n\theta_n}\eta_n) 
= 
\frac{e^{- n\theta_n\,\eta_n^2/2}}{\sqrt{2\pi  n\theta_n}\eta_n}
\left[\sum_{k=0}^{M-1}(-1)^k\frac{A_k}{(\sqrt{n\theta_n}\eta_n)^{2k}} 
+ O\left(n^{-M}\eta_n^{-2M}\right)\right]
\end{equation}
Since $\theta_n/\lambda_n\to 1$ as $n\to\infty$, we have from~\eqref{eq:eta Taylor expansion near 1} that $\lim\limits_{n\to\infty}\eta_n=0$ and so the $N$th error term in the expansion for $\bar{\Phi}$ dominates the $N$th error term in the expansion for $Q$.
Therefore, substituting~\eqref{eq:large argument Phi expansion} with $M=N$ into the above expansion for $Q$ yields the stated result.
\end{proof}

\section{Proof of Theorem~\ref{Thm:mean_variance}}
\label{sec:mean and variance}
We first present a proposition summarising the asymptotic behaviour of $H_n$.
\begin{proposition}
\label{prop:H_n asymptotics}
Fix $\alpha \in\RR$ and $\beta\geq 0$. Define the sequence $(\lambda_n)_{n\in\NN}$ by $\lambda_n = 1 + \alpha n^{-\beta}$. Then as $n\rightarrow \infty$,
\begin{enumerate}
    \item\label{lemma_part:H_n_low_T} If $\beta=0$ and $\alpha\in(-1,0)$, then
\begin{equation*}
\hspace{-5mm}
        \frac{1}{H_n(n\lambda_n)} = \sqrt{\frac{n}{2\pi}}\,
        e^{-n\left[\alpha - \log(1+\alpha)\right]}\,[1+o(1)]
\end{equation*}
    
    \item\label{lemma_part:H_n_low_T_scalingwindow} If $\beta\in(0,1/2)$ and $\alpha<0$, then
\begin{equation*}
\hspace{-5mm}
        \frac{1}{H_n(n\lambda_n)} =\sqrt{\frac{n}{2\pi}}\, \exp{\left(-\frac{\alpha^2}{2}n^{1-2\beta}[1+o(1)]\right)}\,
            [1+o(1)]
 \end{equation*}
    
    \item\label{lemma_part:H_n_beta=1/2} If $\beta=1/2$ and $\alpha\in\RR$, then
\begin{equation*}
\hspace{-5mm}
        \frac{1}{H_n(n\lambda_n)} = \frac{\varphi(\alpha)}{\bar{\Phi}(\alpha)}\,\sqrt{n}
        \,\left[1 + o(1))\right]
\end{equation*}
    
    \item\label{lemma_part:H_n_criticalwindow} If $\beta>1/2$ and $\alpha\in\RR$, then
\begin{equation*}
\hspace{-5mm}
     \frac{1}{H_n(n\lambda_n)} = \sqrt{\frac{2}{\pi}}\,\sqrt{n}\,\left[1 + o(1)\right]
 \end{equation*}
    
    \item\label{lemma_part:H_n_high_T_scalingwindow} If $\beta\in(0,1/2)$ and $\alpha>0$, then
\begin{equation*}
\hspace{-5mm}
     \frac{1}{H_n(n\lambda_n)} =  \alpha n^{1-\beta}\left[1 + \frac{n^{2\beta-1}}{\alpha^2} + \frac{\ind(\beta\leq 1/3)n^{\beta-1}}{\alpha} - \frac{2n^{4\beta-2}}{\alpha^4} + o(n^{4\beta-2})\right]
\end{equation*}
    
    \item\label{lemma_part:H_n_high_T} If $\beta=0$ and $\alpha>0$, then
\begin{equation*}
\hspace{-5mm}
        \frac{1}{H_n(n\lambda_n)} = \alpha n
        \left[1 + \frac{(1+\alpha)}{\alpha^2\, n} 
        - \frac{(1+\alpha)(2+\alpha)}{\alpha^4\, n^2} + O(n^{-3})\right]
\end{equation*}
\end{enumerate}
\end{proposition}

\begin{proof}
Let $\eta_n:=\eta(1+\alpha n^{-\beta})$.
The fundamental difference between the various cases stated above is the behaviour of $\sqrt{n}\eta_n$ as $n\to\infty$. 
If $\beta>0$, then~\eqref{eq:eta Taylor expansion near 1} implies
\begin{equation}
\label{eq:root(n) eta asymptotics}
\sqrt{n}\,\eta_n = \alpha\,n^{1/2 -\beta}[1+O(n^{-\beta})]
\end{equation}
while if $\beta=0$ and $\alpha>-1$, then we have
\begin{equation}
\label{eq:root(n) eta asymptotics when beta=0}
\sqrt{n}\,\eta_n = \eta(1+\alpha)\,n^{1/2}
\end{equation}
with $\eta(1+\alpha)$ positive (negative) iff $\alpha$ is positive (negative).

It follows that for Cases~\eqref{lemma_part:H_n_low_T} and~\eqref{lemma_part:H_n_low_T_scalingwindow} one has 
$\lim\limits_{n\to\infty}\sqrt{n}\eta_n=-\infty$.
Setting $\theta_n=1$ in Lemma~\ref{lem:Temme's expansion}, and noting that $\lim\limits_{x\to-\infty}\bar{\Phi}(x)=1$, we then have 
\begin{equation}
\label{eq:Q(n,lambda n) asymptotics for cases i and ii}
Q(n,n\lambda_n) = 1 + o(1)
\end{equation}
Consequently, from~\eqref{eq:H in terms of eta} and~\eqref{eq:Stirling's formula} we have
$$
\frac{1}{H_n(n\lambda_n)} = \sqrt{\frac{n}{2\pi}}\,e^{-n\,\eta_n^2/2}
[1+o(1)]
$$
Cases~\eqref{lemma_part:H_n_low_T} and~\eqref{lemma_part:H_n_low_T_scalingwindow} then follow by combining with~\eqref{eq:eta definition} and~\eqref{eq:root(n) eta asymptotics}, respectively.

In Cases~\eqref{lemma_part:H_n_beta=1/2} and~\eqref{lemma_part:H_n_criticalwindow}, Eq.~\eqref{eq:root(n) eta asymptotics} shows that
$\sqrt{n}\eta_n$ converges to a finite limit, of $\alpha$ and $0$, respectively. From Lemma~\eqref{lem:Temme's expansion} we then obtain
\begin{eqnarray}
\label{Q(n,lambda n) asymptotics for Cases iii and iv}
\lim\limits_{n\to\infty}Q(n,n\lambda_n) =
\begin{cases}
\bar{\Phi}(\alpha), & \beta=1/2,\\
\bar{\Phi}(0), & \beta>1/2
\end{cases}
\end{eqnarray}
The stated result then follows from~\eqref{eq:H in terms of eta} by applying~\eqref{eq:Stirling's formula}.

In Cases~\eqref{lemma_part:H_n_high_T_scalingwindow} and~\eqref{lemma_part:H_n_high_T}, Eqs.~\eqref{eq:root(n) eta asymptotics} and~\eqref{eq:root(n) eta asymptotics when beta=0} respectively show that $\lim\limits_{n\to\infty}\sqrt{n}\eta_n=+\infty$. Consider first Case~\eqref{lemma_part:H_n_high_T_scalingwindow}. 
Applying Lemma~\ref{lem:Temme's expansion after cancellation} with $\theta_n = 1$ and $N=3$ to~\eqref{eq:H in terms of eta} we obtain
$$
    \frac{n\,H_n(n\lambda_n)}{\Gamma^*(n)} 
    =
    \sum_{k=0}^{2}(-1)^k
    \frac{q_k(\alpha n^{-\beta})}{(\alpha\,n^{-\beta})^{2k+1}\,n^k} 
    + O(n^{7\beta-3}) \nonumber \\
$$
where in obtaining the error term we used $\eta_n\sim \alpha n^{-\beta}$. Using the fact that $q_k$ is a polynomial of degree $2k$, it then follows that
\begin{equation}
\label{Eq:BnH}
    \frac{n\,H_n(n\lambda_n)}{\Gamma^*(n)} 
    =
    \sum_{k=0}^{2}\sum_{l=0}^{2k}
    \frac{(-1)^ka_{k,l}}{\alpha^{2k+1-l}}n^{\psi(k,l)} +  O(n^{7\beta - 3})
\end{equation}
where 
$$
\psi(k,l):=(2k+1-l)\beta-k = \beta-(1-2\beta)k-\beta l
$$

Since $\beta\in(0,1/2)$, the function $\psi(k,l)$ is decreasing with respect to the usual (coordinatewise) partial order on $\NN^2$, and we therefore have $\psi(2,0)>\psi(k,l)$ for all $(k,l)$ with $k\geq 2$ and $(k,l)\neq (2,0)$. 
Therefore, using the explicit values of the $a_{k,l}$ from Lemma~\ref{lem:Temme's expansion after cancellation} we obtain
$$
\sum_{k=0}^{2}\sum_{l=0}^{2k}
\frac{(-1)^ka_{k,l}}{\alpha^{2k+1-l}}n^{\psi(k,l)} \\
 =
 \frac{n^{\psi(0,0)}}{\alpha} - \frac{n^{\psi(1,0)}}{\alpha^3} - \frac{n^{\psi(1,1)}}{\alpha^2} - \frac{n^{\psi(1,2)}}{12\,\alpha} 
 + \frac{3\,n^{\psi(2,0)}}{\alpha^5} + o(n^{\psi(2,0)}) 
$$
Since $\psi$ is decreasing we have $\psi(0,0) >\psi(1,0)>\psi(2,0)$ and $\psi(1,1)>\psi(1,2)$, but the way $\psi(2,0)$ relates to $\psi(1,1)$ and $\psi(1,2)$ depends on the value of $\beta$. Inserting explicit values for $\psi(k,l)$
and discarding all terms which are $o(n^{\psi(2,0)})$ we then obtain
$$
    \frac{n H_n(n\lambda_n)}{\Gamma^*(n)}
    =
    \frac{n^{\beta}}{\alpha} - \frac{n^{3\beta -1}}{\alpha^3} - \frac{\ind(\beta\leq 1/3)n^{2\beta - 1}}{\alpha^2}
 - \frac{\ind(\beta\leq 1/4)n^{\beta - 1}}{12\alpha} +
    \frac{3n^{5\beta - 2}}{\alpha^5} + o(n^{5\beta - 2})
$$
where we have used the fact that $n^{7\beta - 3}=o(n^{5\beta -2})$ for $\beta\in(0,1/2)$. Applying~\eqref{eq:Stirling's formula} and carefully expanding the reciprocal then yields the stated result.

Finally, Case~\eqref{lemma_part:H_n_high_T} follows from 
Lemma~\ref{lem:Tricomi's expansion} by setting $\theta_n =1$, $\lambda = 1 + \alpha$ and $N=3$, which yields
\begin{equation*}
        H_n(n\lambda_n) = \frac{1}{\alpha n}\left[1 - \frac{1+\alpha}{\alpha^2 n} + \frac{2(1+\alpha)^2+(1+\alpha)}{\alpha^4 n^2} + O(n^{-3})\right]
\end{equation*}
\end{proof}

\begin{proof}[Proof of Theorem~\ref{Thm:mean_variance}]
Let $\lambda_n=1+\alpha n^{-\beta}$ with $\alpha,\beta$ as chosen in any of the cases stated in Theorem~\ref{Thm:mean_variance}, and set $z_n=1/\lambda_n$. Then, from~\eqref{eq:mean in terms of H} and~\eqref{eq:var in terms of H} we obtain
\begin{eqnarray}
   \EE_{n,z_n}(\sL_n) &= -\alpha n^{1-\beta} + \frac{1}{H_n(n\lambda_n)} -1\\
   \var_{n,z_n}(\sL_n) &= n 
   + \left(1+\frac{1}{H_n(n\lambda_n)}\right)\alpha\,n^{1-\beta}
    - \frac{1}{H_n^2(n\lambda_n)}
\end{eqnarray}
The stated results then follow directly from Proposition~\ref{prop:H_n asymptotics} (combined with~\eqref{eq:conditional normal moments} in the case of $\beta=1/2$).
\end{proof}

\section{Proof of Theorem~\ref{Thm:CLT}}
\label{sec:limit theorems}
We now turn to the proof of Theorem~\ref{Thm:CLT}.
\begin{proof}[Proof of Theorem~\ref{Thm:CLT}]
We prove the weak convergence results stated in Theorem~\ref{Thm:CLT} by establishing pointwise convergence of the standardised tail distribution. 

Let $\lambda_n=1+\alpha n^{-\beta}$, and set $z_n=1/\lambda_n$. 
Fix $y\in\RR$, and define the sequence $\kappa_n(y)$ via
\begin{equation}
    \kappa_n(y) := \begin{cases} 
    |\alpha|n + y\sqrt{(1+\alpha)n} & \text{if}~\beta = 0~\text{and}~\alpha\in(-1,0) \\[10pt]
    |\alpha|n^{1-\beta} + y\sqrt{n} & \text{if}~\beta\in(0,1/2)~\text{and}~\alpha<0 \\[10pt]
    (y-\alpha)\sqrt{n}          & \text{if}~\beta=1/2,~\text{and}~\alpha\in\RR \\ [10pt]
     y\sqrt{n}                  & \text{if}~\beta>1/2,~\text{and}~\alpha\in\RR\\[10pt]
     y \,n^{\beta}/\alpha       & \text{if}~\beta\in(0,1/2),~\text{and}~\alpha>0\\[10pt]
     y - 1                      & \text{if}~\beta=0,~\text{and } \alpha>0
    \end{cases}
\end{equation}
Now define the sequence $(\mu_n)_{n\in\NN}$ so that $n\mu_n = n - \floor{\kappa_n(y)} - 1$, and observe that
$$
\mu_n = 1 - \frac{\kappa_n(y)}{n} + O(n^{-1})
$$
It follows from~\eqref{eq:eta Taylor expansion near 1} that for Cases~\eqref{CLT:low temp}-\eqref{CLT:critical window} of Theorem~\ref{Thm:CLT} we have
$$
\sqrt{n \mu_n}\eta(\lambda_n/\mu_n) = y + o(1), \qquad n\to\infty
$$
and Lemma~\ref{lem:Temme's expansion} then yields
\begin{equation}
\label{eq:Q(n mu, n lambda) asymptotics for cases i to iv}
   \lim\limits_{n\to\infty} Q(n\mu_n,n\lambda_n) = \bar{\Phi}(y) 
\end{equation}

In Cases~\eqref{CLT:low temp} and~\eqref{CLT:low temp window} we have $\kappa_n(y)\to+\infty$ for all $y\in\RR$, and 
it then follows from~\eqref{eq:tail_distribution}, 
\eqref{eq:Q(n,lambda n) asymptotics for cases i and ii}
and~\eqref{eq:Q(n mu, n lambda) asymptotics for cases i to iv}
that for all $y\in\RR$
$$
\lim\limits_{n\to\infty}
\PP_{n,z_n}(\sL_n>\kappa_n(y)) = \bar{\Phi}(y) = \PP(X>y)
$$
The stated weak convergence result then follows immediately.

Similarly, applying~\eqref{Q(n,lambda n) asymptotics for Cases iii and iv} and~\eqref{eq:Q(n mu, n lambda) asymptotics for cases i to iv} to~\eqref{eq:tail_distribution}, we see that in Case~\eqref{CLT:boundary} we have
$$
\lim\limits_{n\to\infty} 
\PP_{n,z_n}(\sL_n>\kappa_n(y)) 
= \ind(y<\alpha) + \ind(y\ge\alpha)\frac{\bar{\Phi}(y)}{\bar{\Phi}(\alpha)}
=\PP(X>y|X>\alpha)
$$
while in Case~\eqref{CLT:critical window} we have
$$
\lim\limits_{n\to\infty}
\PP_{n,z_n}(\sL_n>\kappa_n(y)) 
= \ind(y<0) + \ind(y\ge0)\frac{\bar{\Phi}(y)}{\bar{\Phi}(0)}
=\PP(X>y|X>0)
$$

We now turn attention to the non-Gaussian limits. Consider Case~\eqref{CLT:high temp window}. 
Since
\begin{eqnarray}
    \frac{\lambda_n}{\mu_n} = \frac{1 + \alpha n^{-\beta}}{1 - yn^{\beta - 1}/\alpha + O(n^{-1})} = 1 + \alpha n^{-\beta} + O(n^{\beta -1})  \nonumber
\end{eqnarray}
we see that in either of the cases $\theta_n=\mu_n$ or $\theta_n=1$, we have
\begin{equation}
\label{eq:xi asymptotics for high temp window CLT}
    \frac{\lambda_n}{\theta_n} = \lambda_n + O(n^{\beta-1})
\end{equation}
and therefore also
\begin{equation}
\label{eq:eta asymptotics for high temp window CLT}
\eta\left(\lambda_n/\theta_n\right) = \alpha n^{-\beta}[1+o(1)]
\end{equation}
Consequently, setting $N=1$ in Lemma~\ref{lem:Temme's expansion after cancellation} and inserting~\eqref{eq:xi asymptotics for high temp window CLT} and~\eqref{eq:eta asymptotics for high temp window CLT} implies that 
$$
    Q(n\theta_n,n\lambda_n) = \frac{e^{-n\theta_n\eta^2(\lambda_n/\theta_n)/2}}{\sqrt{2\pi n\theta_n}}
    \frac{n^{\beta}}{\alpha} \left[1 + O(n^{2\beta - 1})\right]    
$$
and it follows that
\begin{eqnarray*}
    \PP_{n,z_n}(\walklength > \kappa_n) &=&
    \ind(y<0) + 
    \ind(y\ge0)\,\frac{Q(n\mu_n,n\lambda_n)}{Q(n,n\lambda_n)} \nonumber \\ 
    & = &  \exp\left[\frac{n\eta^2(\lambda_n) - n\mu_n\eta^2(\lambda_n/\mu_n)}{2}\right]\left[1+o(1)\right]
\end{eqnarray*}
Moreover, setting $\zeta_n:=\floor{\kappa_n(y)} + 1$, so that $\mu_n = 1-\zeta_n/n$ and $\zeta_n=y n^{\beta}/\alpha +O(1)$, yields
\begin{eqnarray*}
 \hspace{-20mm} \frac{n\eta^2(\lambda_n) - n\mu_n\eta^2(\lambda_n/\mu_n)}{2}
  &=& -\zeta_n\,\log(\lambda_n) - \zeta_n - n\,\log\left(1-\frac{\zeta_n}{n}\right) + \zeta_n\log\left(1-\frac{\zeta_n}{n}\right) \\
&=&
-\zeta_n\,\log(\lambda_n) - \zeta_n - n\,\log\left(1-\frac{\zeta_n}{n}\right) + O(n^{2\beta-1})
\\
&=&
-\zeta_n\,\log(\lambda_n)+ O(n^{2\beta-1})
\\
  &=& -y + o(1)
\end{eqnarray*}
We then conclude that
$$
\lim_{n\to\infty}\PP_{n,z_n}(\sL_n>\kappa_n(y)) = \ind(y<0)+\ind(y\ge0)e^{-y} = \PP(Y>y)
$$

Finally, consider Case~\eqref{CLT:high temp} and let $\lambda:=1+\alpha$. Since
$$
\mu_n=1-\floor{y}/n
$$
we see that in either of the cases $\theta_n=\mu_n$ or $\theta_n=1$, we have
$$
    \lambda - \theta_n = \alpha + O(n^{-1})
$$
and so setting $N=1$ in Lemma~\ref{lem:Tricomi's expansion} yields
$$
Q(n\theta_n,n\lambda) = 
\frac{(n\lambda)^{n\theta_n}}{\Gamma(n\theta_n)}
\frac{e^{-n\lambda}}{\alpha\,n}\left[1+O(n^{-1})\right]
$$
From~\eqref{eq:tail_distribution} we then have 
$$
\PP_{n,z_n}(\sL_n>y-1) =
\ind(y<1) + 
\ind(y\ge1)(1+\alpha)^{-\floor{y}}
\frac{n^{-\floor{y}}\Gamma(n)}{\Gamma(n-\floor{y})}\left(1+O(n^{-1})\right)
$$
However~\eqref{eq:definition of gamma star} and~\eqref{eq:Stirling's formula} imply that
$$
\frac{\Gamma(n)}{\Gamma(n-\floor{y})} \sim n^{\floor{y}}
$$
and so we obtain
$$
\lim_{n\to\infty}\PP_{n,z_n}(\sL_n > y - 1) = 
\ind(y<1) + 
\ind(y\ge1)(1+\alpha)^{-\floor{y}} 
= 
\PP(W_{\alpha} > y)
$$
\end{proof}

\ack{This work was supported under the Australian Research Council's Discovery Projects funding scheme (Project Nos. DP140100559 and DP180100613.) Y. D. acknowledges the support by the National Key R\&D Program of China under Grant No. 2016YFA0301604 and by the National Natural Science Foundation of China under Grant No. 11625522. We thank Gordon Slade for bringing to our attention his recent preprint~\cite{Slade2019}. }

\section*{References}
\bibliographystyle{unsrt}

\end{document}